\documentclass[11 pt]{article} 

\usepackage{amssymb, amsmath, amsthm, bbm, mathrsfs}
\usepackage{fullpage}

\newcommand{\tr}{\operatorname{trace}}

\theoremstyle{definition}
\newtheorem{defn}{Definition}

\newtheorem{lem}[defn]{Lemma}

\theoremstyle{plain}
\newtheorem{thm}[defn]{Theorem}
\newtheorem{prop}[defn]{Proposition}
\newtheorem{cor}[defn]{Corollary}

\begin{document}

\title{A class of Grassmannian fusion frames}
\author{Emily J. King\thanks{This author was supported in part by a fellowship from the Alexander von Humboldt Foundation -- as well as a Department of Education GAANN Fellowship and a University of Maryland Graduate School Ann G. Wylie Dissertation Fellowship.}}

\maketitle

\begin{abstract}Transmitted data may be corrupted by both noise and data loss.  Grassmannian frames are in some sense optimal representations of data transmitted over a noisy channel that may lose some of the transmitted coefficients.  Fusion frame (or frame of subspaces) theory is a new area that has potential to be applied to problems in such fields as distributed sensing and parallel processing.  Grassmannian fusion frames combine elements from both theories.  A simple, novel construction of Grassmannian fusion frames with an extension to Grassmannian fusion frames with local frames shall be presented.  Some connections to sparse representations shall also be discussed. \textbf{Keywords:} fusion frame, Grassmannian packing, Walsh sequences, Hadamard matrices, Grassmannian fusion frame, robustness. \textbf{AMS Classification:} 42C15, 15B34, 14M15
\end{abstract}

\section{Introduction}
\subsection{Motivation}
When data is transmitted over a communication line, the received message may be corrupted by noise and data loss.  As an oversimplified example, if the message \verb+1729+ is sent to you, you could receive the noisy message \verb+1728+ or nothing at all.  Representing data in a way that is resilient to such problems is clearly desirable.  Expressing data using a redundant frame provides some protection, but some frames work better than others.  Grassmannian frames are examples of such superior representations.  Fusion frames are objects which may be applied to a number of different fields and afford representations of higher dimensional measurements, as opposed to the one dimensional measurements of a traditional frame.  We begin by laying a foundation in frame theory in Section \ref{sec:frame}. Then we specialize in Section \ref{sec:grass} and generalize in Section \ref{sec:fus} by introducing Grassmannian frames and fusion frames. The main result, a construction of Grassmannian fusion frames utilizing Hadamard matrices, is in Section \ref{sec:const} and is followed by an extension of the construction to a system robust to local erasures in Section~\ref{sec:local}.  Concluding remarks and a comment about sparsity are in Section \ref{sec:con}.
\subsection{Frames}\label{sec:frame}
Let $\mathbb{F}$ be either $\mathbb{R}$ or $\mathbb{C}$.
\begin{defn} 
A sequence $\lbrace e_{i} \rbrace _{i=1}^N$ in $\mathbb{F}^M$ is a \emph{frame} for $\mathbb{F}^M$ if there exist constants $0<A\leq B < \infty$ such that
\begin{equation} \label{eqn:1}
\forall x \in \mathbb{F}^M\textrm{, } \qquad A\Vert x \Vert ^2 \leq \sum_{i = 1}^N \vert \langle x,e_{i} \rangle \vert ^2 \leq B\Vert x \Vert ^2 {.}
\end{equation}
$A$ and $B$ are called \emph{frame bounds}.  A frame is \emph{tight} if $A$ and $B$ may be chosen so that $A = B$ and \emph{Parseval} if $A=B=1$.  A frame is \emph{unit-norm} if $\Vert e_{i} \Vert = 1$ for $1 \leq i \leq N$.  A frame is \emph{equiangular} if for some $\alpha$, $\vert \langle e_i, e_j \rangle \vert = \alpha$ for all $i \neq j$. 
\end{defn}
Every orthonormal basis is a frame. One may view frames as generalizations of orthonormal bases which mimic the reconstruction properties (i.e.: $\forall x, x= \sum \langle x, e_i \rangle e_i$)  of orthonormal bases but may have some redundancy.   We remark that $\{e_i \}$ is a tight frame with frame bound $A$ if and only if
\begin{equation}\label{eqn:tight}
\forall x \in \mathbb{F}^M\textrm{, } \qquad AI_M x= \sum_{i=1}^N \langle x,e_{i} \rangle e_i, 
\end{equation}
where $I_M$ is the identity matrix for $\mathbb{F}^M$.

\section{Grassmannian fusion frames}

\subsection{Grassmannian frames}\label{sec:grass}
Goyal \emph{et al.}~proved that a unit-norm frame is \emph{optimally robust against} (\emph{o.r.a.}) noise and one erasure if the frame is tight \cite{GKK01}.  Furthermore, a unit-norm frame is o.r.a.~multiple erasures if it is Grassmannian \cite{StH03,BKo06}.
\begin{defn}
Define 
\begin{equation*}
\mathscr{F}(N, \mathbb{F}^M) = \lbrace \lbrace e_i \rbrace_{i=1}^N \subset \mathbb{F}^M :  \textrm{$\lbrace e_i \rbrace$ is a unit-norm frame for $\mathbb{F}^M$}\rbrace.
\end{equation*}
The \emph{maximal frame correlation} is 
\begin{equation}\label{eqn:corr}
\mathcal{M}_\infty (\lbrace e_i \rbrace_{i=1}^N) = \max_{1\leq i < j \leq N} \lbrace \vert \langle e_i, e_j \rangle \vert \rbrace.
\end{equation}
A sequence of unit-norm vectors $\lbrace u_i \rbrace_{i=1}^N \subset \mathbb{F}^M$ is called a \emph{Grassmannian frame} if it is a solution to 
\begin{equation*}
\min_{\lbrace e_i \rbrace \in \mathscr{F}(N, \mathbb{F}^M)} \lbrace \mathcal{M}_\infty (\lbrace e_i \rbrace_{i=1}^N) \rbrace .
\end{equation*}
\end{defn}
If $N = M$, the Grassmannian frames are precisely the orthonormal bases for $\mathbb{F}^M$.  If $N = 3$ and $M = 2$, the $2$-dimensional real vectors representing the cubic roots of unity are a Grassmannian frame.  However, the vectors representing the fourth roots of unity do not form a Grassmannian frame for $N = 4$ and $M = 2$.  Since $|\langle (1,0) , (-1,0) \rangle | = 1$, the fourth roots of unity actually have the largest possible maximal frame correlation.  The following theorem is proven in a number of classical texts, see \cite{StH03} for one proof and citations of other methods.
\begin{thm}
Let $\lbrace e_i \rbrace_{i=1}^N$ be a unit-norm frame for $\mathbb{F}^M$.  Then
\begin{equation}\label{eqn:max_corr}
\mathcal{M}_\infty (\lbrace e_i \rbrace_{i=1}^N ) \geq \sqrt{\frac{N-M}{M(N-1)}}
\end{equation}
\end{thm}
Equality holds in (\ref{eqn:max_corr}) if and only if $\lbrace e_i \rbrace$ is an equiangular tight frame. Thus, equiangular unit-norm frames are automatically Grassmannian frames.  However, equality can only hold for certain $N$ and $M$.  Bodmann and Paulsen proved a functorial equivalence between real equiangular frames and $\alpha$-regular $2$-graphs, where $\alpha$ depends on $N$ and $M$ \cite{BP05}.  An $\alpha$-regular $2$-graph is a particular type of hypergraph.  This correspondence can be used to characterize when equiangular frames exist.  Other than the case $N = M +1$, there are very few known pairs $(N, M)$ which yield equiangular frames.  Further, there are many pairs $(N,M)$ for which it has been proven that no equiangular frames exist.  When equiangular frames do not exist, it can be complicated to construct Grassmannian frames.
\begin{defn}
For $1 \leq m \leq M$, set $G(M,m)$ to be the collection of $m$ dimensional subspaces of $\mathbb{F}^M$. 
$G(M,m)$ is called a \emph{Grassmannian}.  $G(M,m)$ may be endowed with many mathematical structures, but we shall only be concerned with the metric space structure induced by the \emph{chordal distance} (see, \cite{GrassPack})
\begin{equation*}
\operatorname{dist}(\mathcal{W}_i,\mathcal{W}_j) = [m - \tr(P_i P_j)]^{1/2},
\end{equation*}
for $\mathcal{W}_i, \mathcal{W}_j \in G(M,m)$, where $P_i$ is the orthogonal projection onto $\mathcal{W}_i$.  
\end{defn}
The \emph{Grassmannian packing problem} is the problem of finding $N$ elements in $G(M,m)$ so that the minimal distance between any two of them is as large as possible.  A numerical approach to solving this problem may be found in \cite{DHST08} and a list of best-known packings is posted on \cite{Sloane}.  Finding Grassmannian frames is equivalent to solving the Grassmannian packing problem for $G(M,1)$ restricted to unit-norm frames.  It is natural to ask what analytic structures one obtains by considering the Grassmannian packing problem for $m > 1$. 

We may also consider other geometric relationships among elements of a given Grassmannian which also generalize equiangularity.  The following definition is found in \cite{BCPST} and is equivalent to alternate definitions found in \cite{BEt07} and \cite{LemSei73}.
\begin{defn}
Two $m$ dimensional subspaces $\mathcal{W}_i, \mathcal{W}_{j} \subseteq \mathbb{F}^M$ are \emph{isoclinic} with parameter $\lambda$ if the angle $\theta$ between any $f \in \mathcal{W}_i$ and its orthogonal projection $P_{j}f$ in $\mathcal{W}_{j}$ is unique with $\cos^2(\theta) = \lambda$.  A collection of subspaces is \emph{equi-isoclinic} if they are pairwise isoclinic with the same parameter $\lambda$.
\end{defn} 
\subsection{Fusion frames}\label{sec:fus}
\emph{Fusion frames}, originally called \emph{frames of subspaces}, were introduced by Casazza and Kutyniok in \cite{CaK04}.  There are many potentially exciting applications of fusion frames in areas such as coding theory \cite{Bod07}, distributed sensing \cite{CKL08}, and neurology \cite{GJR06}.  Please see Chapter 13 of \cite{CaKBook} for an overview of the state of the art in fusion frames.
\begin{defn}
A \emph{fusion frame (with respect to $\lbrace \nu_i \rbrace_{i=1}^N$)} for $\mathbb{F}^M$ is a finite collection of subspaces $\lbrace \mathcal{W}_i \rbrace_{i=1}^N$ in $\mathbb{F}^M$ and positive weights $\lbrace \nu_i \rbrace_{i=1}^N$ such that there exist $0 < A \leq B < \infty$ satisfying
\begin{equation*}
A\Vert x \Vert^2 \leq \sum_{i=1}^N \nu_i^2\Vert P_i x \Vert^2 \leq B \Vert x \Vert^2,
\end{equation*}
where $P_i$ is an orthogonal projection onto $\mathcal{W}_i$. If $A=B$, we say that the fusion frame is \emph{tight}. If $A=B=1$, then the fusion frame is called \emph{Parseval}.  Parseval fusion frames are referred to as \emph{weighted projective resolutions of the identity} in \cite{Bod07}.

A tight fusion frame consisting of equal dimensional subspaces with equal pairwise chordal distances is an \emph{equi-distance tight fusion frame}.
\end{defn}
Similar to (\ref{eqn:tight}), a fusion frame is tight with bound $A$ if and only if 
\begin{equation}\label{eqn:tight2}
\sum_{i=1}^N \nu_i P_i = A I_M.
\end{equation}
For the remainder of this paper, we will assume that our weights $\nu_i$ are all equal to $1$ (as in done in \cite{GrassFus} and elsewhere).  Fusion frames may either be viewed as generalizations of frames or special types of frames.  In the former sense, we are merely replacing the projections (modulo constant multiples) of vectors  $x \in \mathbb{F}^M$ (on line (\ref{eqn:1})) onto the subspace spanned by each frame vector with projections onto spaces of dimensions possibly higher than $1$.  In the latter sense, we may see a fusion frame as a frame with subcollections of frame vectors which group in \emph{nice} ways.  As is common in the literature, \emph{nice} is an oversimplification of some very deep properties.  In fact, splitting frames into such sub-collections is related to the (in)famous Feichtinger Conjecture \cite{CKS08}.  We shall make use of the following simple lemma in later proofs.
\begin{lem} \label{lem:tight}
Assume that $\lbrace \mathcal{W}_i \rbrace_{i=1}^N$ is a collection of subspaces in $\mathbb{F}^M$ with corresponding orthonormal bases $\lbrace e^i_j \rbrace_{j=1}^{m_i}$.  Then $\lbrace \mathcal{W}_i \rbrace_{i=1}^N$ is a tight fusion frame if
\begin{equation*}
L^t = \left( e^1_1 \vert e^1_2 \vert \cdots \vert e^1_{m_1} \vert e^2_1 \vert \cdots \vert e^{N}_{m_N}  \right)
\end{equation*}
has equal-norm orthogonal rows.
\end{lem}
\begin{proof}
It suffices to show that (\ref{eqn:tight2}) holds.  Label the rows of $L^t$ as $\{ f_1, f_2, \hdots, f_M \}$ and the rows of  $\left( e^i_1 \vert e^i_2 \vert \cdots \vert e^i_{m_i} \right)$ as $\lbrace f^i_1, f^i_2, \cdots, f^i_M \rbrace$
Then
\begin{eqnarray*}
P_i & = & \left( e^i_{1} \vert e^i_2 \vert \cdots \vert e^i_{m_i} \right)\left( e^i_{1} \vert e^i_2 \vert \cdots \vert e^i_{m_i} \right)^t  \\
& = &  \left( f^i_{1} \vert f^i_2 \vert \cdots \vert f^i_M \right)^t\left( f^i_{1} \vert f^i_2 \vert \cdots \vert f^i_M \right) \\
& = & \left( \langle f^i_k , f^i_\ell \rangle \right)_{k,\ell},
\end{eqnarray*}
and
\begin{eqnarray*}
\sum_{i=1}^N P_i & = & \left(  \sum_{i=1}^N  \langle f^i_k , f^i_\ell \rangle \right)_{k,\ell} = \left( \langle f_k, f_\ell \rangle \right)_{k,\ell}.
\end{eqnarray*}
\end{proof}

\subsection{Robustness} 
As mentioned in the introduction, data may be corrupted by noise and/or erasures.  In order to measure the resilience of a fusion frame representation of data against noise and erasures, one may use a deterministic or a stochastic signal model.  The former case was analyzed in \cite{Bod07} and the latter in \cite{GrassFus}.  The specific models are explained in the respective papers.  We simply summarize the results concerning optimally robust fusion frames. With a deterministic signal model, a  equi-dimensional tight fusion frame is o.r.a.~one subspace erasure if the weights $\nu_i$ are equal and o.r.a.~two subspace erasures if it is equi-isoclinic.   Furthermore, we are interested in when the original signal can be recovered exactly when the measurement on one subspace is lost or corrupted.
\begin{defn}
Let $\{ \mathcal{W}_i \}_{i = 1}^N$ be a tight fusion frame for $\mathbb{F}^M$ with $\dim \mathcal{W}_i = m$ for each $1 \leq i \leq N$.  For each $1 \leq i \leq N$, fix an orthonormal basis for $\mathcal{W}_i$, $\{e^i_k \}_{k=1}^m$.  Then  $\{ \mathcal{W}_i \}_{i = 1}^N$ is \emph{($1$-subspace-loss) correctible} if for any $1 \leq j \leq N$, the mapping 
\begin{equation*}
x \mapsto \oplus_{i = 1}^N (1-\delta_{ij}) P_i x
\end{equation*}
has a left inverse. 
\end{defn}
With the notation used above, an equi-isoclinic fusion frame is $1$-subspace-loss correctible if $mN > M$ and $N > 1$ \cite{Bod07}.  With a stochastic signal model, a fusion frame is o.r.a.~against noise if it is tight, a tight fusion frame is o.r.a.~one subspace erasure if the dimensions of the subspaces are equal, and a tight fusion frame is o.r.a.~multiple subspace erasures if the subspaces are maximally equi-distance.   A maximal equi-distance tight fusion frame is a solution to the Grassmannian packing problem \cite{GrassFus}.  Motivated by this fact, we introduce the following definition which does not explicitly exist in the literature but is simply a succinct label for objects discussed in \cite{GrassFus}.
\begin{defn}
A fusion frame $\lbrace \mathcal{W}_i \rbrace_{i=1}^N$ for $\mathbb{F}^M$ consisting of $d$-dimensional subspaces shall be called a \emph{Grassmannian fusion frame} if it is a solution to the Grassmannian packing problem for $N$ points in $G(M,m)$.
\end{defn}
Note that an equi-isoclinic fusion frame is equi-(chordal)distance and thus Grassmannian.  In \cite{Bod07}, equi-isoclinic Parseval fusion frames are called $2$-\emph{uniform}, parallel to the terminology used for Grassmannian frames in \cite{BP05}.

We would like to construct fusion frames which are robust under either model to noise and multiple erasures.  The idea behind the construction is very simple and makes use of Hadamard matrices.
\section{Hadamard construction}\label{sec:const}
\subsection{Hadamard matrices}
The first Hadamard matrices were discovered by Sylvester in 1867 \cite{Syl}.  In 1893, Hadamard first defined and started to characterize Hadamard matrices, which have the maximal possible determinant amongst matrices with entries from $\{\pm1\}$ \cite{Had}.
\begin{defn}
A \emph{Hadamard matrix} of order $n$ is an $n \times n$ matrix $H$ with entries from $\{\pm 1\}$ such that $HH^t=nI_n$.
\end{defn}
In the original paper by Hadamard, it was proven that Hadamard matrices must have order equal to $2$ or a multiple of $4$.  It is still an open conjecture as to whether Hadamard matrices exist for every dimension divisible by $4$.  However, there are many constructions of Hadamard matrices, which use methods from number theory, group cohomology and other areas of math.  Horadam's book \cite{HadMat} is an excellent resource.  One class of Hadamard matrices are formed from Walsh functions.
\begin{defn}
The \emph{Walsh functions} $\omega_j : [0,1] \rightarrow \{\pm 1\} \subset \mathbb{R}$, $j \geq 0$, are piecewise constant functions which have the following properties:
\begin{enumerate}
\item $\omega_j (0) = 1$ for all $j \geq 0$,
\item $\omega_j$ has precisely $j$ sign changes (zero crossings), and
\item $\langle \omega_j, \omega_k \rangle = \delta_{jk}$.
\end{enumerate}
\end{defn}
Walsh functions have been used for over 100 years by communications engineers to minimize cross talk.  The first four Walsh functions are
\begin{eqnarray*}
\omega_0 & = & \mathbbm{1}_{[0,1]}, \\
\omega_1 & = & \mathbbm{1}_{[0,1/2)} -\mathbbm{1}_{[1/2,1]}, \\
\omega_2 & = & \mathbbm{1}_{[0,1/4)} -\mathbbm{1}_{[1/4,3/4)}+\mathbbm{1}_{[3/4,1]}, \textrm{ and} \\
\omega_3 & = & \mathbbm{1}_{[0,1/4)} -\mathbbm{1}_{[1/4,1/2)}+\mathbbm{1}_{[1/2,3/4)} -\mathbbm{1}_{[3/4,1]}, 
\end{eqnarray*}
where $\mathbbm{1}_A$ is the function that takes the value $1$ on $A$ and $0$ off of $A$, $A$ measurable.  
By sampling the first $2^n$ Walsh functions at the points $\frac{k}{2^n}$, one obtains a Hadamard matrix; that is,
\begin{equation*}
W_n = \left(\omega_j \left(\frac{k}{2^n}\right)\right)_{0 \leq j,k < 2^n}
\end{equation*}
is a $2^n \times 2^n$ Hadamard matrix.  There is a speedy algorithm, the Fast Hadamard (or Walsh) Transform, for multiplying a vector by such a matrix.
\subsection{Construction}
The idea behind the new construction of the Grassmannian fusion frames is very simple.  We remove the first $j$ rows of a Hadamard matrix $H$ to obtain a submatrix $H'$.  The columns of $H'$ will then be partitioned into spanning sets for subspaces.  Since the elements of the matrix are $\pm1$, the computations should be simplified and the resulting (fusion) frame should be easy to implement.  This construction can be seen as similar to Naimark Complementation \cite{CCHK09}, although we start out knowing the larger Hilbert space.  It is well-known (see, for example, \cite{HadMat}) that any Hadamard matrix can be normalized so that the first row consists solely of $1$s.  
\begin{prop}
Let $H$ be an $n \times n$ Hadamard matrix indexed by $0, \hdots , n  -1$ which has been normalized so that the first row consists solely of $1$s.  Then 
\begin{equation*}
\lbrace e_i = \frac{1}{\sqrt{n -1}}(H(i,k))_{1 \leq j \leq n -1} : 0 \leq i \leq n -1 \rbrace
\end{equation*}
is a unit-norm Grassmannian frame for $\mathbb{F}^{n-1}$ with frame bound $\frac{n}{n-1}$.  In particular, the frame is equiangular.
\end{prop}
\begin{proof}
For any $0 \leq i_1, i_2 \leq n -1$, 
\begin{eqnarray*}
\langle H(j,i_1))_{0 \leq j \leq n -1}, H(j,i_2))_{0 \leq j \leq n -1}\rangle & = & n \delta_{i_1,i_2} \textrm{ and}\\
\langle H(i_1,j))_{0 \leq j \leq n -1}, H(i_2,j))_{0 \leq j \leq 2^n -1}\rangle & = & n\delta_{i_1,i_2}.
\end{eqnarray*}
Thus $\langle e_{i_1} , e_{i_2} \rangle = \frac{1}{n -1} (n\delta_{i_1,i_2} -1)$, and the collection is equiangular.  We now show that the $e_i$ do indeed form a frame. Let $x \in \mathbb{F}^{n -1}$ be arbitrary.  We verify that (\ref{eqn:tight}) holds.  Let 
\begin{equation*}
L = \left(\frac{1}{\sqrt{n -1}}(H(j,i)) \right)_{0\leq i \leq n-1, 1 \leq j \leq n-1}.
\end{equation*}
Then, by the orthogonality of the columns,
\begin{equation*}
\sum_{i=0}^{n-1} \langle x , e_i \rangle e_i = L^t L x = \frac{n}{n -1} x.
\end{equation*}
\end{proof}
We now present the main construction.
\begin{thm}\label{thm:grass}
Let $W_n$ be the $2^n \times 2^n$ Walsh-Hadamard matrix indexed by $0, \hdots , 2^n  -1$.  Then 
\begin{equation*}
\left\{ \mathcal{W}_i = \operatorname{span}_{0 \leq k \leq 2^m - 1} \left\{ (W_n(j,i + k2^{n-m}))_{2^m \leq j \leq 2^n-1} \right\} \right\}_{0 \leq i \leq 2^{n -m}- 1}
\end{equation*}
is a tight Grassmannian fusion frame for $\mathbb{F}^{2^n-2^m}$ with frame bound $\frac{2^n}{2^n - 2^m}$ consisting of $2^{n-m}$ $2^m$-dimensional subspaces which is equi-isoclinic with parameter $\lambda = \left( \frac{2^m}{2^n - 2^m}\right)^2$.
\end{thm}
\begin{proof}
Note that the first $2^m$ Walsh functions are constants over intervals of the form $[\frac{\ell}{2^m},\frac{\ell + 1}{2^m} )$, $\ell \in \mathbb{Z}$.  Since the $2^n \times 2^n$ Walsh-Hadamard matrix is generated by sampling at points of the form $\frac{\ell}{2^n}$, the submatrix of $W_n$ formed from the first $2^m$ rows consists of the first  column of $W_m$, repeated $2^{n-m}$ times, followed by the second column of $W_m$, repeated $2^{n-m}$ times, and so on.  Thus, for $0 \leq i_1, i_2 \leq 2^{n-m} -1$,
\begin{equation*}
\langle (W_n(j, i_1 + k_1 2^{n-m}))_{0 \leq j \leq 2^m - 1}, (W_n(j, i_2 + k_2 2^{n-m}))_{0 \leq j \leq 2^m - 1} \rangle = 2^m \delta_{k_1, k_2},  
\end{equation*}
which implies
\begin{eqnarray}
&&\langle (W_n(j, i_1 + k_1 2^{n-m}))_{2^m \leq j \leq 2^n - 1}, (W_n(j, i_2 + k_2 2^{n-m}))_{2^m \leq j \leq 2^n - 1} \rangle \nonumber \\
&&\quad = 2^n \delta_{i_1,i_2}\delta_{k_1,k_2} - 2^m \delta_{k_1, k_2}.  \label{eqn:innpro}
\end{eqnarray}
Set $e^i_k = \frac{1}{\sqrt{2^n - 2^m}}(W_n(j, i + k 2^{n-m}))_{2^m \leq j \leq 2^n - 1}$.  For each $0 \leq i \leq 2^{n-m} - 1$, $\{ e^i_k \}_{0 \leq k \leq 2^m - 1}$
is an orthonormal basis for $\mathcal{W}_i$.  It follows from Lemma~\ref{lem:tight} that $\{ \mathcal{W}_i\}_{i=0}^{2^{n-m}-1}$ is a tight fusion frame with frame bound  $\frac{2^n}{2^n - 2^m}$.  For each $0 \leq i \leq 2^{n-m} - 1$ define
\begin{equation*}
L_i = \left( e^{i}_{0} \vert e^{i}_1 \vert \cdots \vert e^{i}_{2^m-1} \right)^t. 
\end{equation*}
Then for $i_1 \neq i_2$, the trace of $P_{i_1}P_{i_2} $ is
\begin{eqnarray*}
&& \tr L_{i_1}^t L_{i_1} L_{i_2}^t L_{i_2} \\
&& = \tr  L_{i_1}^t  \left(\frac{-2^m}{2^n - 2^m} I_{2^m} \right) L_{i_2}\\
&& = \frac{-2^m}{2^n - 2^m} \sum_{k=0}^{2^m -1 } \langle e_k^{i_1}, e_k^{i_2} \rangle  =  \frac{-2^m}{2^n - 2^m} \left( 2^m \frac{-2^m}{2^n - 2^m} \right)  =  \frac{2^m}{(2^{n-m} - 1)^2}.
\end{eqnarray*}
Thus between any two distinct $\mathcal{W}_i$, the chordal distance is $(2^m - \frac{2^m}{(2^{n-m} - 1)^2})^{1/2}$, and the $\mathcal{W}_i$ form a Grassmannian fusion frame.

All that remains to be shown is that $\{\mathcal{W}_i\}$ is equi-isoclinic.  For an arbitrary subspace $\mathcal{W}_i$, choose an arbitrary element $e_k^i$ of the orthonormal basis $\{e_{\ell}^i\}_{0\leq\ell\leq 2^m - 1}$.  We would like to show that for any $j \neq i$
\begin{equation*}
\frac{\langle e_k^i, P_j e_k^i \rangle }{\Vert e_k^i \Vert \Vert P_j e_k^i \Vert} = \frac{2^m}{2^n - 2^m},
\end{equation*}
independently of $i$, $j$, and $k$.  Using (\ref{eqn:innpro}), we calculate
\begin{equation*}
P_j e_k^i = \sum_{\ell = 0}^{2^m - 1} \langle e_k^i, e_{\ell}^j \rangle e_{\ell}^j = \sum_{\ell = 0}^{2^m - 1} \frac{-2^m}{2^n - 2^m} \delta_{k,\ell} e_{\ell}^j = \frac{-2^m}{2^n - 2^m}  e_k^j .
\end{equation*}
Thus, $\Vert P_j e_k^i \Vert = \frac{2^m}{2^n - 2^m}$, and it follows from (\ref{eqn:innpro}) that
\begin{equation*}
\langle e_k^i, P_j e_k^i \rangle = \frac{-2^m}{2^n - 2^m} \langle e_k^i, e_k^j \rangle = \left( \frac{2^m}{2^n - 2^m} \right)^2.
\end{equation*}
The claim is proven.
\end{proof}
Note that if $n = m+1$, the isoclinic parameter is $\lambda = 1$; that is, the fusion frame is trivial in that it simply contains two copies of the entire space.

We have the following corollary.
\begin{cor}\label{cor:erase}
Let  $\{\mathcal{W}_i\}_{0 \leq i \leq 2^{n-m}-1}$ be constructed as in Theorem~\ref{thm:grass} with $n > m+1$.  Then $\{\mathcal{W}_i\}_{0 \leq i \leq 2^{n-m}-1}$ is 1-subspace-loss correctible.
\end{cor}
\begin{proof}
The fusion frames are equip-isoclinic.  By \cite{Bod07}, they are 1-subspace-loss correctible.
\end{proof}

Unfortunately, it follows from Theorem 3.2 of \cite{CaK08} and Theorem 24 of \cite{Bod07} that any tight fusion frame for $\mathbb{F}^{2^n-2^m}$ with frame bound $\frac{2^n}{2^n - 2^m}$ consisting of $2^{n-m}$ $2^m$-dimensional subspaces is not $k$-subspace-loss correctible for $k \geq 2$.

It should be noted that although \cite{GrassFus} contains Grassmannian fusion frames constructed using Hadamard matrices, the way in which they are used here is completely different.  Specifically, in Example 5.1 of \cite{GrassFus}, the authors construct a Grassmannian fusion frame consisting of $3$-dimensional subspaces in in $\mathbb{R}^7$ using a $4 \times 4$ Hadamard matrix.  The construction makes use of the fact that $7$ is a prime and $3$ is the number of non-zero quadratic residues of $7$.  The construction presented in this paper creates fusion frames for spaces of dimension $2^n - 2^m$ of subspaces of dimension $2^m$.  It follows from elementary number theory that if $2^n - 2^m$ is a prime, then $n = 2, m = 1$ or $n > 1, m \neq 0$. Either way, $2^m$ would not be the number of non-zero quadratic residues.  Thus, the methods are mutually exclusive.

\section{Erasure of Local Frame Vectors}\label{sec:local}
When using a fusion frame to make measurements, one is recording a collection of multi-dimensional data, i.e. $\{ P_i x \}_i$.  Thus, a potential problem is that instead of an entire subspace being lost, only part of the data from a particular $P_i x$ is lost.  Section 4 of \cite{CaK08} deals with this sort of robustness. A natural way to introduce robustness is to represent $P_i x$ not by its coefficients with respect to an orthonormal basis of $\mathcal{W}_i$ but rather a frame of $\mathcal{W}_i$.  We call such a frame a \emph{local frame}.

In this section, we will scale our tight fusion frames with weights $v_i = 1$ to be Parseval fusion frames with weights $v_i \neq 1$.  We will also use the notation $\{f_{i,k}\}_{k=1}^{K_i}$ to denote a frame of a subspace $\mathcal{W}_i$.  Then we speak of a fusion frame system with local frames: 
\[
\{ \mathcal{W}_i, v_i, \{f_{i,k}\}_{k=1}^{K_i}\}_{i=1}^N.
\]

We now present a definition and theorem from \cite{CaK08}.
\begin{defn}
Let $W = \{ \mathcal{W}_i, v_i, \{f_{i,k}\}_{k=1}^{K_i}\}_{i=1}^N$ be a Parseval fusion frame system with local Parseval frames.  For each $1 \leq i \leq N$, define
\[
L_{F_i} = \left( f_{i,1} \vert f_{i,2} \vert \cdots \vert f_{i,K_i} \right)^t
\]
(the \emph{analysis operator} for $\{f_{i,k}\}_{k=1}^{K_i}$). Similarly, for each $1 \leq i \leq N$ and $1 \leq \hat{k} \leq K_i$, define
\[
L_{F_i}^{\hat{k}} = \left( f_{i,1} \vert f_{i,2} \vert \cdots  \vert f_{i,\hat{k}-1} \vert f_{i,\hat{k}+1}\vert\cdots \vert f_{i,K_i} \right)^t
\]
Then we define the associated \emph{1-erasure reconstruction error} 
\[
\mathcal{E}_1(W) = \max \left\{ \bigg\Vert Id - \left( (L_{F_{\hat{i}}}^{\hat{k}})^\ast  L_{F_{\hat{i}}}^{\hat{k}}+ \sum_{i=1, i\neq \hat{i}}^N v_i^2 L_{F_i}^\ast L_{F_i}\right) \bigg\Vert : 1 \leq \hat{i} \leq N, 1 \leq \hat{k} \leq K_{\hat{i}}\right\}.
\]
\end{defn}
\begin{thm}
Let $W = \{ \mathcal{W}_i, v_i, \{f_{i,k}\}_{k=1}^{K_i}\}_{i=1}^N$ be a Parseval fusion frame system with local Parseval frames.  Then the following conditions are equivalent.
\begin{itemize}
\item The Parseval frame system $W$ satisfies $\mathcal{E}_1(W) = \min \big\{\mathcal{E}_1 (\tilde{W}): \tilde{W} = \{ \tilde{\mathcal{W}}_i, v_i, \{f_{i,k}\}_{k=1}^{K_i}\}_{i=1}^N$ is a Parseval fusion frame system with local Parseval frames satisfying $\dim \tilde{W}_i = \dim W_i$ for all $1 \leq i \leq N.\big\}$.
\item We have
\[
\Vert f_{i,k} \Vert^2 = \frac{\dim W_i}{K_i} \textrm{ for all $1 \leq i \leq N$, $1 \leq k \leq K_i$}.
\]
\end{itemize}
\end{thm}
This theorem shows what should be intuitive, namely, that among Parseval fusion frame systems with local Parseval fusion frames which have same subspace dimensions and size of local frames, the ideal systems have local frames which are optimally robust to one erasure in the sense of frames \cite{GKK01}.  

Thus, to obtain a fusion frame system which is robust against both subspace erasures and local erasures, one could start with an equi-isoclinic Parseval fusion frame $\{ \mathcal{W}_i, v_i \}_{i=1}^N$ and choose an equal-norm Parseval frame $\{f_{i,k}\}_{k=1}^{K_i}$ for each $\mathcal{W}_i$.  In fact, one only needs to choose one equal-norm Parseval frame.

We first present a trivial corollary to Theorem~\ref{thm:grass}.
\begin{cor}\label{cor:grass}
Let $W_n$ be the $2^n \times 2^n$ Walsh-Hadamard matrix indexed by $0, \hdots , 2^n  -1$.  Then 
\begin{equation*}
\left\{ \mathcal{W}_i = \operatorname{span}_{0 \leq k \leq 2^m - 1} \left\{ (W_n(j,i + k2^{n-m}))_{2^m \leq j \leq 2^n-1} \right\}, \sqrt{\frac{2^n-2^m}{2^n}} \right\}_{0 \leq i \leq 2^{n -m}- 1}
\end{equation*}
is a Parseval Grassmannian fusion frame for $\mathbb{F}^{2^n-2^m}$  consisting of $2^{n-m}$ $2^m$-dimensional subspaces which is equi-isoclinic with parameter $\lambda = \left( \frac{2^m}{2^n - 2^m}\right)^2$.
\end{cor}
We now present our straightforward construction of a robust Parseval fusion frame with a local Parseval fusion frame.
\begin{prop}
Let $\{\mathcal{W}_i,v_i\}_i$ be as in Corollary~\ref{cor:grass}.  Fix $M > 2^m$.  Choose a Parseval equal-norm frame $\{ f_\ell \}_{\ell = 1}^{M}$ for $\mathbb{F}^{2^m}$.  Set
\[
L_F = \left( f_{1} \vert f_{2} \vert \cdots \vert f_{M} \right)^t,
\]
\[
e^i_k = \frac{1}{\sqrt{2^n - 2^m}}(W_n(j, i + k 2^{n-m}))_{2^m \leq j \leq 2^n - 1},
\]
and
\[
L_i = \left( e^{i}_{0} \vert e^{i}_1 \vert \cdots \vert e^{i}_{2^m-1} \right)^t. 
\]
Finally, for each $1 \leq i \leq N$ and $1 \leq \ell \leq M$, set $f_{i,\ell}$ to be the $\ell$th row of $L_F L_i$.  Then $ W = \{ \mathcal{W}_i, v_i, \{f_{i,\ell}\}_{\ell=1}^{M}\}_{i=1}^N$ is $1$-subspace-loss correctible and minimizes the local $1$-erasure reconstruction error.
\end{prop}
\begin{proof}
Since the new system has the same subspaces as the fusion frame constructed in Theorem~\ref{thm:grass}, $W$ is $1$-subspace-loss correctible by Corollary~\ref{cor:erase}.  We now must show that for each $1 \leq i \leq N$, $\{f_{i,\ell}\}_{\ell=1}^{M}$ is an equal-norm Parseval frame for $\mathcal{W}_i$.  Fix $i$.  To show it is a Parseval frame, we note that for $x \in \mathcal{W}_i$,
\[
\sum_{\ell = 1}^M \langle x, f_{i,\ell} \rangle  =  L_i^\ast L_F^\ast L_F L_i x = L_i^\ast L_i x = P_i x = x,  
\]
as desired.  Because $\{f_\ell\}_{\ell = 1}^M$ is an equal-norm frame, the columns of $L_F$ are equal-norm.  Furthermore, the rows of $L_i$ are orthonormal by construction.  So the rows of $L_F L_i$, i.e., $\{f_{i,\ell}\}$, are linear combinations of orthonormal vectors by coefficients which square-sum to $\frac{2^m}{M}$.  It follows from the Pythagorean theorem that $\Vert f_{i,\ell_1}\Vert = \Vert f_{i,\ell_2} \Vert$ for all $1 \leq \ell_1, \ell_2 \leq M$.  Since $i$ was arbitrary, we are done.
\end{proof}
There are many papers which address the construction of equal-norm Parseval frames such as $\{ f_\ell \}_{\ell = 1}^{M}$, see, for example \cite{CaKo03,CaKBook}.

\section{Conclusion}\label{sec:con}
By utilizing well-known functions, we were able to construct equi-isoclinic, tight Grassmannian fusion frames in a very straightforward manner.  These particular fusion frames are optimally resilient against noise and erasures under a deterministic or a stochastic signal model.  The construction also automatically yields orthonormal bases for the fusion frame subspaces which have very simple coordinates, namely $\{ \pm 1\}$, which are easier to implement.  The theory of Grassmannian fusion frames has exciting potential.  We hope that the theory will expand beyond the equi-distance constructions presented in this paper and in \cite{GrassFus} to packings in spaces which can not yield equi-distance fusion frames.  The construction is then simply extended to yield a Grassmannian fusion frame with local Parseval equal-norm frames.

We also conclude with a quick comment about sparsity.  \emph{Sparse representations} and \emph{compressed sensing} are both terms that have exploded in popularity. The basic premise seems unbelievable at the surface: that it is possible to get a ``correct'' solution to an underdetermined system.  It ends up that the same properties that make Grassmannian (fusion) frames robust to noise and erasures also make them systems that yield optimal sparseness in the following manner.  Suppose we form a matrix $A$ with columns consisting of the vectors of a (redundant) frame, and we would like to find the sparsest $x$ which solves the the underdetermined equation $Ax = b$ for some non-zero $b$.  This problem is NP-hard.  We might try convex relaxation, instead searching for the solution $x$ with minimal $\ell^1$ norm.  However, we would like some guarantee that we would get a desirable solution.  One measure of success is called the \emph{mutual coherence} of the matrix $A$, which is precisely Equation \ref{eqn:corr} \cite{BDE09}.  Low coherence guarantees that sparse solutions to the original problem are unique and that this relaxation will work \cite{DE03,DET06,GN03,Tem00}.  This work has been generalized to fusion frames, with a focus on finding representations in as few as possible subspaces $\mathcal{W}_i$.  Minimizing the \emph{frame coherence} is equivalent to solving the Grassmannian packing problem \cite{BKR10}.  Thus the systems constructed in this paper may be viewed not only as robust but yielding optimal sparseness.

\bibliography{thesis}{}
\bibliographystyle{amsalpha}

\end{document}